\documentclass[11pt,a4paper]{article}

\usepackage[T1]{fontenc}
\usepackage{amsmath,amssymb,amsfonts,latexsym,epsfig}
\usepackage{amssymb,amsmath}
\usepackage{amsfonts}

\usepackage{psfrag,xspace}
\usepackage{multirow}

\usepackage{algorithmic} 
\usepackage[boxruled,boxed,algo2e]{algorithm2e}
\usepackage{hyperref}
\usepackage[mathbf]{euler}
\usepackage{notations}
\usepackage{url}
\usepackage{vector}

\usepackage{fullpage}

\newtheorem{theorem}{Theorem}[section]
\newtheorem{lemma}[theorem]{Lemma}

\newtheorem{proposition}[theorem]{Proposition}
\newtheorem{corollary}[theorem]{Corollary}

\newtheorem{conjecture}[theorem]{Conjecture}
\newtheorem{remark}[theorem]{Remark}

\begin{document}

\title{Random polynomials and expected
complexity of bisection methods for real solving
}

\author{
  Ioannis Z. Emiris%
  \thanks{National Kapodistrian University of Athens, Greece.
    Email: \texttt{emiris(at)di.uoa.gr}}
  \and 
  Andr\'e Galligo%
  \thanks{University of Nice, France.
    Emai: \texttt{galligo(at)unice.fr}}
  \and 
   Elias P.~Tsigaridas%
   \thanks{University of Athens, Greece 
     and \AA{}rhus University, Denmark.
     Email: \texttt{elias.tsigaridas(at)gmail.com}}
}

\date{}
\maketitle

\begin{abstract}
  Our probabilistic analysis sheds light to the following questions:
  Why do random polynomials seem to have few, and well separated real roots,
  on the average?
  Why do exact algorithms for real root isolation may
  perform comparatively well or even better than numerical ones?

  We exploit results by Kac, and by Edelman and Kostlan
  in order to estimate the real root separation of degree $d$ polynomials with
  i.i.d.\ coefficients that follow two zero-mean normal distributions: 
  for $SO(2)$ polynomials, the $i$-th coefficient has variance ${d \choose i}$,
  whereas for Weyl polynomials its variance is ${1/i!}$.
  By applying results from statistical physics,
  we obtain the expected (bit) complexity of \func{sturm} solver,
  $\sOB(r d^2 \tau)$, where $r$ is the number of real roots 
  and $\tau$ the maximum coefficient bitsize.
  Our bounds are two orders of magnitude tighter than the record worst
  case ones.
  We also derive an output-sensitive bound in the worst case.
   
  The second part of the paper shows that the expected number of real roots
  of a degree $d$ polynomial in the Bernstein basis is $\sqrt{2d}\pm\OO(1)$,
  when the coefficients are i.i.d.\ variables with moderate standard deviation.
  Our paper concludes with experimental results which corroborate our
  analysis.
\end{abstract}

\vspace{0.9mm}
\noindent
{\bf Categories and Subject Descriptors:}
F.2 [Theory of Computation]: Analysis of Algorithms and Problem Complexity;
I.1 [Computing Methodology]: Symbolic and algebraic manipulation: Algorithms


\section{Introduction}

One of the most important procedures in computer algebra and algebraic
algorithms is root isolation of univariate polynomials. The goal is
to compute intervals in the real case, or squares in the
complex case, that isolate the roots of the polynomial and to compute
one such interval, or square, for every root.

We restrict ourselves to exact algorithms, i.e.\ algorithms that perform
arithmetic with rational numbers of arbitrary size.  
The best known algorithms are subdivision algorithms, 
based on Sturm sequences (\func{sturm}),
or on Descartes' rule of sign (\func{descartes}),
or on Descartes' rule and the Bernstein basis representation (\func{bernstein}).
Subdivision algorithms mimic binary search and their
complexity depends on separation bounds.
They are given an initial interval, or compute one containing all real roots.
Then, they repeatedly subdivide it until it is certified that zero or
one real root is contained in the tested interval.

Thanks to important recent progress
\cite{Dav:TR:85,Yap:SturmBound:05,ESY:descartes,emt-lncs-2006}, 
the complexity of \func{sturm}, \func{descartes} and \func{bernstein}
is, in the worst case, $\sOB(d^4 \tau^2)$,
where $d$ is the degree of the polynomial and $\tau$ the maximum
coefficient bitsize.
The bound holds even when the polynomial is non-squarefree,
and we also compute (all) the multiplicities.
This requires a preprocessing of complexity $\sOB(d^2\tau)$,
in order to compute the square-free factorization.
The new polynomial has coefficients of size $\OO(d+\tau)$.
The complexity of this stage, although significant in practice,
is asymptotically dominated.  
In this paper we consider the behavior of \func{sturm} on
random polynomials of various forms.
Our results can be extended to \func{descartes} and \func{bernstein}.

Another important exact solver (\func{cf}) is based
on the continued fractions expansion of the real roots
e.g.~\cite{Akritas:implementation,sharma-tcs-2008,et-tcs-2007}.
Several variants of this solver exist, depending on the method used to
compute the partial quotients of the real roots.  Assuming
the Gauss-Kuzmin distribution holds for the real algebraic
numbers, it was proven \cite{et-tcs-2007},
that the expected complexity is $\sOB( d^4 \tau^2)$.
By spreading the roots, the expected complexity becomes
$\sOB( d^3 \tau)$ \cite{et-tcs-2007}.
The currently known worst-case bound is $\sOB(d^4 \tau^2)$ \cite{mr-jsc-2009}.
This paper reduces the gap between \func{sturm} \func{cf}.

Numerical algorithms compute an approximation, up to a desired accuracy,
of all complex roots. They can be turned
into isolation algorithms by requiring the accuracy to be equal 
to the theoretical worst-case separation bound.
The current record is $\sOB( d^3 \tau)$ and is achieved by
recursively splitting the polynomial until one obtains linear
factors that approximate sufficiently the roots \cite{Sch82,Pan02jsc}.
It seems that the bounds could be improved to $\sOB(d^2 \tau)$ with a more
sophisticated splitting process.
We should mention that optimal numerical algorithms are very difficult to implement.

Even though the complexity bounds of the exact algorithms
are worse than those of the numerical ones, recent implementations of the
former tend to be competitive, if not superior, in practice, e.g.\
\cite{JohKraLynRicRus-issac-06,RouZim:solve:03,emt-lncs-2006,et-tcs-2007}. 
Our work attempts to provide an explanation for this. 
There is a huge amount of work concerning root isolation
and the references stated represent only the tip of the iceberg;
we encourage the reader to refer to the references.

Most of the work on random polynomials, which typically concerns
polynomials in the monomial basis, focuses on the {\em number} of real roots.
Kac's \cite{Kac-BAMS-43} celebrated result
estimated the expected number of real roots of 
random polynomials (named after himself)
as $\frac{2}{\pi}\log{d} + \OO(1)$,
when the coefficients are standard normals i.i.d.\ or uniformly distributed,
and $d$ is the degree of the polynomial.
We refer the reader to
e.g.~\cite{bb-lms-1932,lo-lms-1938,ErdTur-AM-1950} for a historical
perspective and to \cite{BhSa86} for various references.
A geometric interpretation of this result and many generalizations
appear in \cite{EdeKos-BAMS-95}. 
We mainly examine $SO(2)$ polynomials, where the $i$-th coefficient is 
an i.i.d.\ Gaussian random variable of zero mean and variance ${d \choose i}$.
According to \cite{EdeKos-BAMS-95},
they are ``the most natural definition of random polynomials'', 
see also \cite{ss-bezout-2}.
Their expected number of real roots is $\sqrt{d}$.
For Weyl polynomials, the $i$-th coefficient is an i.i.d.\ Gaussian random
variable of zero mean and variance ${1/i!}$, and
the expected number of real roots is about
$\frac{2}{\pi} \sqrt{d} + \OO(1)$ where higher-order terms are not known to date
\cite{sm-jsp-2008}.
For results on complex roots we refer to 
e.g.~\cite{hannay-jopa-1996,hammersley-1954}.  

Our first contribution concerns the expected bit complexity
of \func{sturm}, when the input is random polynomials with i.i.d.\
coefficients; notice that their roots are {\em not} independently distributed!
In other words, we have to go beyond the theory of Kac, and Edelman and Kostlan,
in order to study the statistical behavior of root differences and,
more precisely, the {\em minimum} absolute difference.
We examine $SO(2)$ and Weyl random polynomials, and exploit
the relevant progress achieved in statistical physics.
In fact, these polynomial classes are of particular interest in statistical
physics because they model zero-crossings in diffusion equations and,
eventually, a chaotic spin wave-function 
\cite{bd-jsp-1997,hannay-jopa-1996,sm-jsp-2008}. 
The key observation is that, by applying these results, we can quantify the
correlation between the roots, which is sufficiently weak, but does
exist. For both classes of polynomials we prove an expected case bit
complexity bound of $\sOB(r \, d^2 \tau)$, where $r$ is the number of
real roots.
A close related bound was speculated in \cite{Johnson-phd-91}, based
on experimental evidence.

Our bounds are tighter than those of the worst case by two factors.
In the course of this analysis, \func{sturm} is shown to be output-sensitive,
with complexity proportional to the number of real roots in the given interval,
even in the worst case. A similar bound appeared in \cite{heindel-jacm-1971}.

Besides polynomials in the monomial basis, polynomials in the
Bernstein basis are important in many applications, e.g.\ CAGD and
geometric modeling.  They are of the form 
$\sum_{i=0}^{d}{a_i {d\choose i}x^i(1-x)^{d-i}}$. 
For the random polynomials that we
consider, $a_i$ are standard normals i.i.d. random variables, that is
Gaussians with zero mean and variance one. Such polynomials are also
important in Brownian motion \cite{kowalski-amm-2006}.  In
\cite{ad-joc-2009}, they examine random polynomial systems; they also
estimate the expected number of real roots of a polynomial in the
Bernstein basis as $\sqrt{d}$, when the variance is ${d \choose i}$.
This left open the case, see also \cite{kowalski-amm-2006}, 
of smaller variance, that is polynomial and
not exponential in $d$.

Our second contribution is to examine
random polynomials in the Bernstein basis of degree $d$, 
with i.i.d.\ coefficients with mean zero and ``moderate'' variance 
$\Theta(1/ \sqrt{d/(i(d-i))})$, for $d>i>0$.
Indeed, we have $1 \geq  \sqrt{d/(i(d-i))} \geq 2/\sqrt{\pi d}$.
We prove that the expected number of real roots of these polynomials is
$\sqrt{2d}\pm\OO(1)$.
We conclude  with experimental results which corroborate our analysis,
and shows that these polynomials behave like polynomial with variance~1.  
This is the first step towards bounding the expected complexity of
solving polynomials in the Bernstein basis.

The rest of the paper is structured as follows.
First we specify our notation.
Sec.~\ref{sec:subdivision-solvers} and 
\ref{sec:average-complexity}
applies our expected-case analysis to estimating the
real root separation bound, and to estimating the complexity of \func{sturm} solver.
Sec.~\ref{sec:random-bernstein} determines the expected number of real roots of
random polynomial in the Bernstein basis and supports our bounds by experimental results.
The paper concludes with a discussion of open questions.

\textbf{Notation.}
\OB means bit complexity and the \sOB-notation means that we
are ignoring logarithmic factors.
For $A = \sum_{i=1}^{d}{a_i X^i} \in \ZZ[X]$, $\dg{A}$ denotes its degree.
\bitsize{A} denotes an upper bound on the bitsize of the
coefficients of $A$ (including a bit for the sign).
For $\rat{a} \in \QQ$, $\bitsize{ \rat{a}}\ge 1$ is 
the maximum bitsize of the numerator and the denominator.
$\Delta$ is the separation bound of $A$, 
that is the smallest distance between two (real or complex, depending on the context) roots of $A$.

\section{Subdivision-based solvers} \label{sec:subdivision-solvers}

In order to make the presentation self-contained, we present in some detail the
general scheme of the subdivision-based solvers.
The pseudo-code of a such a solver is found in Alg.~\ref{alg:Subdivision}.
Our exposition follows closely \cite{emt-lncs-2006}.

The input is a square-free polynomial $A \in \ZZ[x]$ 
and an interval $\interval{I}_0$, that contains the real roots of $A$
which we wish to isolate;  usually it contains all the positive real roots of $A$.  
In what follows, except if explicitly stated otherwise,
we consider only the roots (real and/or complex) of $A$ with
positive real part, since similar results could be obtained for roots with
negative real part  using the transformation $x \mapsto -x$.
Our goal is to compute rational numbers between the real roots of $A$ in
$\interval{I}_0$.

\begin{myalgorithm}[tbp] \scriptsize \dontprintsemicolon \linesnumbered
  \SetFuncSty{textsc} \SetKw{RET}{{\sc return}} \SetKw{OUT}{{\sc output \ }} 
  \SetVline \KwIn{Square-free $A\in \ZZ[x]$, $\interval{I}_0= [ 0, \rat{B}]$} 
  \KwOut{A list of isolating intervals for the real roots of $A$ in $\interval{I}_0$}
  
  \BlankLine

  $\FuncSty{initialization}_{\FuncSty{sm}}( A, \interval{\interval{I}}_0)$ \;
  
  $L \leftarrow \emptyset,\, Q \leftarrow \emptyset$, 
  $Q \leftarrow \FuncSty{push}( Q, \{A, \interval{I}_0\})$ \;
  
  \While{ $Q \neq \emptyset$}{
    \nllabel{alg:Subdivision-while-loop}
    
    $\{f, \interval{I}\} \leftarrow \FuncSty{pop}( Q)$\;
    $V \leftarrow \FuncSty{count}_{\FuncSty{sm}}( f, \interval{I})$ \;
    
    \Switch{ $V$ }{
      
      \lCase{ $V = 0$ }{ \KwSty{continue}\; }
      \lCase{ $V = 1$ }{ $L \leftarrow \FuncSty{ add}( L, \interval{I})$ \; }
      \Case{ $V > 1$ } {
        $\{f_L, \interval{I}_L\}, \{f_R, \interval{I}_R\} \leftarrow 
        \FuncSty{split}_{\FuncSty{sm}}( f, \interval{I})$ \;
        $Q \leftarrow \FuncSty{push}( Q, \{f_L, \interval{I}_L\})$,
        $Q \leftarrow \FuncSty{push}( Q, \{f_R, \interval{I}_R\})$ \;
      }
    }
  }
  \RET $L$ \;
  \caption{ $\func{subdivisionSolver}(A, \interval{I}_0)$ \newline}
  \label{alg:Subdivision}
\end{myalgorithm}

The algorithm uses a stack $Q$ that contains pairs of the form 
$\{f,\interval{I} \}$. The semantics are that we want to
isolate the real roots of $f$ contained in interval $\interval{I}$.
$\func{push}(Q, \{f, \interval{I}\})$
inserts the pair $\{f, \interval{I}\}$ to the top of stack $Q$
and $\func{pop}(Q)$ returns the pair at the top of the stack and
deletes it from $Q$.
$\func{add}(L, \interval{I})$ inserts 
$\interval{I}$ to the list $L$ of the isolating intervals.

There are 3 sub-algorithms with index \func{sm}, which have different
specializations with respect to the subdivision method applied,
namely \func{sturm}, \func{descartes}, or \func{bernstein}.
Generally, $\func{initialization}_{\func{sm}}$ does the necessary pre-processing, 
$\func{count}_{\func{sm}}(f, \interval{I})$ returns the number
(or an upper bound) of the real roots of $f$ in $\interval{I}$,
and $\func{split}_{\func{sm}}(f, \interval{I})$ splits
$\interval{I}$ to two equal subintervals and possibly modifies $f$.

The complexity of the algorithm depends on the number of times the {while}-loop
(Line~\ref{alg:Subdivision-while-loop} of Alg.~\ref{alg:Subdivision}) 
is executed and on the cost of $\func{count}_{\func{sm}}(f, \interval{I})$ and
$\func{split}_{\func{sm}}(f,\interval{I})$.
At every step, since we split the tested interval 
to two equal sub-intervals, we may assume that the bitsize of
the endpoints is augmented by one bit.
If we assume that the endpoints of $\interval{I}_0$ have bitsize $\tau$,
then at step $h$, the bitsize of the endpoints of
$\interval{I} \subseteq \interval{J}_0$ is $\tau + h$.  

Let $n$ be the number of roots with positive real part,
and $r$ the number of positive real roots, so $r\leq n \leq d$.
Let the roots with positive real part, be
$\alpha_j = \Re(\alpha_j) + \ii \, \Im(\alpha_j)$, where $1 \leq j \leq n$
and the index denotes an ordering on the real parts.
Let $\Delta_{i}$ be the smallest distance between $\alpha_i$ and another
root of $A$, and $s_i = \bitsize{\Delta_i}$.
Finally, let the separation bound, i.e.\ the smallest distance between two
(possibly complex) roots of $A$ be $\Delta$ and its bitsize be 
$s= \bitsize{\Delta}$.

\subsection{Upper root bound}
\label{sec:bound}

Before applying a subdivision-based algorithm, we should compute a bound,
$\rat{B}$, on the (positive) roots. We will express this bound as a function of the bitsize
of the separation bound and the degree of the polynomial.
There are various bounds for the  roots of a polynomial,
e.g.\ \cite{Yap2000,Hong:jsc:98,MignotteStefanecu},
and references therein.
For our analysis we use the following bound \cite{Hong:jsc:98} on the positive
real parts of the roots,  
$ \rat{ B} = \Floor{ 2 \max_{a_i < 0} \min_{a_k>0, k >i}{ \Abs{ \frac{a_i}{a_d}}^{1/(k-i)}}}$,
for which we have the estimation \cite{Hong:jsc:98,sharma-tcs-2008} 
$\alpha_r \leq \Re(\alpha_{n}) < \rat{ B} < \frac{8d}{\ln{2}} \Re(\alpha_{n})$.
The bound can be computed in $\sOB( d^2 \tau)$.

If we multiply the polynomial by $x$, then 0 is a root.
By definition of $s$, we have $|\log(|\Re(\alpha_i) - \Re(\alpha_j)|)| \leq s$,
for any $i \not= j$. 
Hence, we have the following inequalities 
\begin{displaymath}
  \begin{array}{rclcrr}
    \Re(\alpha_1) & - & 0 & \leq & 2^{s} \\
    \Re(\alpha_2) & - & \Re(\alpha_1) & \leq & 2^{s} \\
                 &&       &\vdots& \\
    \Re(\alpha_{n-1}) & - & \Re(\alpha_{n-2}) & \leq & 2^{s} \\
    \Re(\alpha_{n})   & - & \Re(\alpha_{n-1}) & \leq &  2^{s}&(+)\\ \hline
     && \Re(\alpha_n) & \leq & n \, 2^{s}
  \end{array}
\end{displaymath}
Thus, we have
$  \rat{ B} < \frac{8d}{\ln{2}} \Re(\alpha_{n}) < \frac{8d}{\ln{2}} n 2^{s}
  < 16 \, d^2\, 2^{s} < d^2 \, 2^{4 + s}$.
Hence, we can deduce that $\bitsize{ \rat{ B}} = \OO( s + \lg{d})$.

\begin{lemma}
  \label{lem:bound}
  Let $A \in \ZZ[x]$, where $\dg{ A} = d$ and $\bitsize{ A} = \tau$.
  We can compute a bound, $\rat{ B}$, on the positive real parts of the roots of $A$, for
  which it holds $\rat{ B} < d^2 \, 2^{4 + s}$, and $\bitsize{ \rat{ B}} = \OO( s + \lg{d})$.
\end{lemma}

\begin{remark}
  In the worst case, the asymptotics of, more or less, 
  all root bounds in the literature,
  e.g.\ \cite{Yap2000,Hong:jsc:98,MignotteStefanecu}, 
  are same,  since 
  $\rat{ B} \leq \max_{i}{ \abs{ a_i}} \leq 2^{\tau}$,
  and $\bitsize{ \rat{ B}} \leq \tau$.
  However, it is very important in practice to have good initial bounds.
  Good initial estimations of the roots can speed up the
  implementation by 20\%  \cite{Krandick:Isolation}. 
\end{remark}

\section{On expected complexity} 
\label{sec:average-complexity}

Expected complexity aims to capture and quantify the property for an algorithm to 
be fast for most inputs and slow for some rare instances of these inputs. 
Let $E$ denote the set of inputs, and assume it is equipped with a probability 
measure $\mu$; then let $c(I)$ denote the usual worst-case 
complexity of the considered algorithm for input $I$. 
By definition, the expected complexity is the integral $\int_E c(I) \mu(I)$. 

In our setting the set $E$ depends on a parameter $d$ (the degree of the input 
polynomial), and we are interested in the asymptotic expected
complexity when $d$ tends to infinity. 
Each $E_d$ is equipped with a probability measure $\mu_d$ (also called 
distribution) of the sequence of the (normalized) coefficients of the input 
polynomial and we consider the cases where there exists a limit distribution. 

\subsection{Strategy and Independence} 
A natural strategy is to decompose $E_d$ into two subsets $G_d$ and $R_d$ (G 
stands for generic and R for rare), such that $c(I)$ is small for $I\in G_d$ 
while $\mu_d(I)$ is very small for $I\in R_d$ and moreover the two partial 
integrals $\int_{G_d} c(I) \mu_d(I)$ and $\int_{R_d} c(I) \mu_d(I)$ are balanced or 
at least both small.

We face another difficulty. Classical properties and estimates in 
Probability theory are often expressed for a sequence of independent variables 
(i.i.d.) but most natural bijective transformations performed in Computer 
Algebra do not respect independence. For instance, if $X$ and $Y$ are independent 
random variables, then $U:=X+Y$ and $V:=X-Y$ are not independent.
In our setting, even if we consider a model of distribution of coefficients which 
assumes that they are i.i.d., then this does not imply that the roots are 
i.i.d.\ and we cannot apply usual tools or estimates. 
However, as we are interested in asymptotic behavior, for some models of 
distribution of coefficients it happens that the limit distribution of the 
roots behave almost like a set of independent variables, i.e.\ they
have very weak correlation. So we can invoke general 
classical estimates for our analysis.

When this is not the case, a useful tool is the two-point, or multi-point, 
correlation function. They express the defect of 
independence between a set of random variables and classically serve,
e.g., to compute standard deviations.

Hereafter, we restrict ourselves to models of 
distribution of coefficients, hence induced distribution of roots, for which the 
corresponding probability measures and correlation functions have already been 
studied. Hopefully these models will provide good approximations for the 
situations encountered in the many applications.  

\subsection{$SO(2)$ polynomials}
\label{sec:so2-polys}

We consider the univariate polynomial $A = \sum_{i=0}^d{a_i x^i}$, the
coefficients of which are i.i.d. normals with mean zero and variances 
${d  \choose i}$, where $0 \leq i \leq d$.
Alternatively, we could consider $A$ as $A = \sum_{i=0}^{d}\sqrt{d  \choose i} \,a_i \,x^i$,
where $a_i$ are i.i.d. standard normals.
These polynomials are considered by Edelman and Kostlan \cite{EdeKos-BAMS-95}
to be ``the more natural definition of a random polynomial''.
They are called $SO(2)$  because the joint probability
distribution of their zeros is $SO(2)$ invariant, after homogenization.
In \cite{sm-jsp-2008} they are called {\em binomial}.
Let $\rho(t) = \frac{\sqrt{d}}{\pi(1+t^2)}$ be the {true density function},
i.e.\ the expected number of real zeros per unit length at a point $t \in \RR$.   
The expected number $r$ of real roots of $A$ is given by
$r = \int_{\RR}{\rho(t) dt} = \sqrt{d}$ \cite{EdeKos-BAMS-95}.  
Let $\alpha_j$ be the real roots of $A$ in their natural ordering, where $1 \leq j \leq r$.

We define the straightened zeros of $A$ as
\begin{displaymath}
  \zeta_j = \mathcal{P}(\alpha_j) = \sqrt{d} \; \arctan(\alpha_j) / \pi,
\, j=1,\dots,r,
\end{displaymath}
in bijective correspondence with the real roots $\alpha_j$ of the random polynomial,
where $\mathcal{P}(t) = \int_0^{t}{\rho(u)\; du}$.
Moreover, the ordering is preserved.
The straightened zeros are uniformly distributed on the circle of length $2\sqrt{d}$
\cite[sec.5]{bd-jsp-1997}.
This is a strong property and implies that the joint probability distribution density
function of two, resp.\ $m$, (distinct) straightened zeros coincides with their 2-point,
resp.\ $m$-point, correlation function \cite{bd-jsp-1997}.

\begin{proposition} {\rm \cite[Thm.~5.1]{bd-jsp-1997}}
  Following the previous notation,
  as $d \rightarrow \infty$ the limit 2-point correlation of the
  straightened zeros is
  $k(s_1,s_2)\rightarrow \pi^2 |s_1-s_2| /4$, when $s_1-s_2 \rightarrow 0$.
\end{proposition}

Let $\Delta(\alpha) = \min_{1 \leq i <r}\{\alpha_{i+1} - \alpha_{i}\}$
and $\Delta(\zeta) = \min_{1 \leq i <r}\{\zeta_{i+1} - \zeta_{i}\}$
be the separation bound of the real roots of $A$ and the straightened
zeros, respectively.
We consider each straightened zero uniformly distributed on a
straight-line interval of length $2\sqrt{d}$. For two such zeros, we
can consider one horizontal and one vertical such interval, defining
a square, which represents their joint probability space.
Since the real roots are naturally ordered, if
two of them lie in a given infinitesimal interval, they must be consecutive.

Let $Z$ be a zone bounded above and below by a diagonal at vertical
distance $l$ from the main diagonal of the unit square.  The
probability $\Pr[\Delta(\zeta) \leq l]$ that there exist two zeros
lying in a given interval of infinitesimal length $l$ tends to the
integral of $k(s_1,s_2)$ over the straightened zeros lying in $Z$, as
$d\rightarrow\infty$:
\begin{displaymath}
  \begin{array}{rl}
    \Pr[\Delta(\zeta) \leq l] \rightarrow & \int_{Z} k(s_1,s_2) \,ds_1\,ds_2  \\
    =&2 \int_0^{2\sqrt{d}} ds_1\, \int_{s_1}^{s_1+l}  k(s_1,s_2) \,ds_2  \\
    =&\frac{\pi^2}{2} \int_0^{2\sqrt{d}} ds_1\, \int_{s_1}^{s_1+l} \abs{s_1 - s_2} \,ds_2=
    \frac{\pi^2 \, \sqrt{d}}{2} \, l^2 ,
  \end{array}
\end{displaymath}
where the first integral is over all straightened zeros, which lie in an
interval of size $2\sqrt{d}$.  
Notice that $k(s_1, s_2)$ is essentially the joint probability density function of
two real roots.
Using Markov's inequality, e.g.~\cite{papoulis-book-1991} we have 
$\Pr[\Delta(\zeta) \geq l] \leq \E[\Delta(\zeta)]/l$, so 
\begin{displaymath}
  \E[\Delta(\zeta)] \geq l \,\Pr[\Delta(\zeta) \geq l]
  = l - l \,\Pr[\Delta(\zeta) < l] > l - \frac{\pi^2 \, \sqrt{d}}{2} \, l^3 .
\end{displaymath} 
This bounds the asymptotic expected separation conditioned on the hypothesis that
it tends to zero, as $d\rightarrow\infty$.
If we choose $l = 1/(d^{c} \tau)$, where $c\geq 1$ is a (small) constant,
which is in accordance with the assumption of $l\rightarrow 0$, 
then 
$\E[\Delta(\zeta)] > \tfrac{1}{d^{c} \tau} -\tfrac{\pi^2}{2 \,d^{3c-1/2} \tau^3}$.

\vspace{-10pt}
\begin{displaymath}
  \E[\Delta(\zeta)] = \E[ \min_{1\leq i < r}\{ \zeta_{i+1} - \zeta_i\}] =
\end{displaymath} 
\vspace{-19pt}
\begin{displaymath}
  \frac{\sqrt{d}}{\pi} \,\E[ \min_{1\leq i < r}\{ \arctan( \alpha_i) - \arctan(\alpha_{i+1}) ] =
\end{displaymath}
\vspace{-15pt}
\begin{displaymath}
  \frac{\sqrt{d}}{\pi} \,
  \E[ \min_{1\leq i < r}\{ \arctan\Paren{\frac{\alpha_i- \alpha_{i+1}}{1+\alpha_i\, \alpha_{i+1}}} \}]  
  > \frac{1}{d^{c} \tau} - \frac{\pi^2}{2 \, d^{3c-1/2} \tau^3} \Leftrightarrow
\end{displaymath} 
\vspace{-7pt}
\begin{displaymath}
    \E[ \min_{1\leq i < r}\{ \arctan\Paren{\frac{\alpha_i- \alpha_{i+1}}{1+\alpha_i\, \alpha_{i+1}}} \}]  
  > \frac{\pi}{d^{c+1/2} \tau} - \frac{\pi^3}{2 \, d^{3c} \tau^3} .
\end{displaymath}
Function $\arctan$ is strongly monotone, and $1 + \alpha_i \alpha_{i+1} \geq 1$,
for all $i$, except where $\alpha_i$ is the largest negative root and
$\alpha_{i+1}$ is the smallest positive root. But we can
treat this case separately, since zero is an obvious separation point. 
\begin{eqnarray*}
  \E[ \min_{1\leq i < r}\{ \alpha_i- \alpha_{i+1} \}]  
   \geq  
   \E[ \min_{1\leq i < r}\{ \Paren{\frac{\alpha_i-\alpha_{i+1}}{1+\alpha_i\, \alpha_{i+1}}} \}] >
   \\ >   
   \tan( \frac{\pi}{d^{c+1/2} \tau} - \frac{\pi^3}{2 \, d^{3c} \tau^3})
  \geq
  \frac{\pi}{d^{c+1/2} \tau} - \frac{\pi^3}{2 \, d^{3c} \tau^3} ,
\end{eqnarray*}
where the latter inequality follows from the series expansion
$\tan x=x+x^3/3+\cdots$ for $x\in(0,\pi/2)$.

\begin{lemma}
  \label{lem:real-avg-sep-s02}
  Let $A \in \ZZ[x]$ of degree $d$, the coefficients of which are
  i.i.d.\ variables that follow a normal distribution with variances
  ${d \choose i}$, then for
  the expected value of the separation bound of the real roots it holds 
  $\E[\Delta] > \frac{\pi}{d^{c+1/2} \tau} - \frac{\pi^3}{2 \, d^{3c} \tau^3}$,
  for a constant $c\geq 1$,
  and 
  $\E[s] = \E[\bitsize{\Delta}] = \OO(\lg{d} + \lg{\tau})$.
\end{lemma}

\subsection{Weyl polynomials}
\label{sec:weyl-polys}

We consider random polynomials, known as Weyl polynomials, 
which are of the form 
$$
A = \sum_{i=0}^{d}{ a_i x^i / \sqrt{i!}},
$$
where the coefficients $a_i$ are independent standard normals. 
Alternatively, we could consider $A$
as $A = \sum_{i=0}^{d}{ a_i x^i}$, where $a_i$ are normals of mean
zero and variance $1/\sqrt{i!}$. 
The density of the real roots of Weyl polynomials is 
\begin{displaymath}
  \rho(t) = \frac{1}{\pi} \sqrt{
    1 + \frac{t^{2d}(t^2-d-1)}{ e^{t^2} \Gamma(n+1,t^2)} -  
    \frac{t^{4d+2}}{ (e^{t^2} \Gamma(n+1,t^2))^2}
    },
\end{displaymath}
where $\Gamma$ is the incomplete gamma function.
The expected number of real roots is
$r = \int_{\RR}{ \rho(t) dt} \sim \frac{2}{\pi} \sqrt{d}$ \cite{sm-jsp-2008},  
where the higher order terms of the number of real roots are not
explicitly known up to now.

The asymptotic density, for $d \rightarrow \infty$, is
\begin{equation}
  \rho(t) =
  \begin{cases}
    \pi^{-1},              &  \abs{t} \ll \sqrt{d} \\
    \frac{d}{\pi t^{2}},   & \abs{t} \gg \sqrt{d} 
  \end{cases}
  \label{eq:weyl-asympt-rho}
\end{equation}

A useful observation is that the density of the real roots of the
Weyl polynomials is similar to the density of the real eigenvalues of
Ginibre random matrices, that is $d \times d$ matrices with elements
Gaussian i.i.d. random variables \cite{EdeKos-BAMS-95,sm-jsp-2008}.

We consider only the real zeros of $A$ that are inside the disc
centered at the origin with radius  $\sqrt{d}$ since outside the disc
there is only a constant number of them.
In this case the density is represented by the first branch of
(\ref{eq:weyl-asympt-rho}). 

We work as in the case of the $SO(2)$ polynomials.
Now $\mathcal{P}(t) = \int_0^{t}{ \rho(u) du} = t/\pi$.
The straightened zeros, $\zeta_i$, are given by 
\begin{displaymath} 
  \zeta_i = \mathcal{P}( \alpha_i) = \alpha_i / \pi,
\end{displaymath}
and they are uniformly distributed in $[0, \sqrt{d}/\pi]$ {\rm \cite{sm-jsp-2008}}.
The joint probability distribution density
function of two straightened zeros coincides with their 2-point correlation function.

\begin{proposition} {\rm \cite{sm-jsp-2008}}
  Under the previous notation,
  as $d \rightarrow \infty$ the limit 2-point correlation of the
  straightened zeros is
  $w(s_1,s_2)\rightarrow |s_1-s_2| /(4 \pi)$, when $s_1-s_2 \rightarrow 0$.
\end{proposition}

Working as in the case of the $SO(2)$ polynomials, the probability
$\Pr[\Delta(\zeta) \leq l]$ that there exist two roots lying in  
a given interval of infinitesimal length $l$ tends to the integral of $w(s_1,s_2)$
over the straightened zeros lying in $Z$, as $d\rightarrow\infty$:
\begin{displaymath}
  \begin{array}{rl}
    \Pr[ \Delta(\zeta) \leq l]  = & \int_Z{ w(s_1,s_2) ds_1 ds_2} \\
    =&\int_0^{\sqrt{d}/\pi} \int_{s_1-l}^{s_1+l}  w(s_1,s_2) \,ds_1\,ds_2
    = \frac{l^2 \sqrt{d}}{4 \pi^2},
  \end{array}
\end{displaymath}
and using Markov's inequality
\begin{displaymath}
  \Pr[\Delta(\zeta) \geq l] \leq \E[\Delta(\zeta)]/l 
  \Longleftrightarrow 
  \E[\Delta] > l - \frac{\sqrt{d}}{4 \pi^2} l^3.
\end{displaymath}
If we choose $l= 1/( d^{c} \tau)$, where $c \geq 1$ is a (small) constant,
we get 
$\E[\Delta(\zeta)] > \frac{1}{d^{c} \tau} - \frac{1}{4\pi^2 d^{3c-1/2} \tau^3}$
and 
$\E[\Delta(\alpha)] > \frac{\pi}{d^{c} \tau} - \frac{1}{4\pi d^{3c-1/2} \tau^3}$.

\begin{lemma} \label{lem:real-avg-sep-weyl}
  Let $A \in \ZZ[x]$ of degree $d$, the coefficients of which are
  i.i.d.\ variables that follow a normal distribution with variances
  $1/i!$, then for the expected value of the separation bound of the
  real roots it holds 
  $\E[\Delta] > \frac{\pi}{d^{c} \tau} - \frac{1}{4\pi d^{3c-1/2} \tau^3}$
  and $\E[s] = \E[\bitsize{\Delta}] = \OO(\lg{d} + \lg{\tau})$.
\end{lemma}

\subsection{The {\sc sturm} solver}
\label{sec:sturm}

Probably the first certified subdivision-based algorithm 
is the algorithm by Sturm, circa 1835, based on his theorem:
In order to count the number of real roots of a polynomial in an
interval, one evaluates a negative polynomial remainder sequence 
of the polynomial and its derivative over the left endpoint of the interval and counts
the number of sign variations. We do the same for the right endpoint;
the difference of sign variations is the number of real roots.

We assume that the positive real roots are contained in
$[0, \rat{ B}]$ (Sec.~\ref{sec:bound}).
If there are $r$ of them, then we need to compute $r-1$ separating points.
The magnitude of the separation points is at most
$\frac{1}{2}\Delta_j$, for $1 \leq j \leq r$, 
and to compute each we need $\Ceil{ \lg{ \frac{2\, \rat{ B}}{\Delta_j}}}$
subdivisions, performing binary search in the initial interval.
Let $T$ be the binary tree that corresponds to the execution of the algorithm
and $\#( T)$ be the number of its nodes, 
or in other words the total number of subdivisions:
\begin{equation}
  \#(T) =  \sum_{j=1}^{r}{ \Ceil{ \lg{ \frac{ 2 \rat{ B}}{\Delta_j}}}} 
  \le 2 r + r \lg{ \rat{ B}} - \sum_{j=1}^{r}{ \lg{\Delta_j}} .
  \label{eq:nb_of_steps}
\end{equation}
Using Lem.~\ref{lem:bound}, we deduce that
$\#(T) = \OO( rs + r \lg(d))$.

The Sturm sequence should be evaluated over a rational number, the
bitsize of which is at most the bitsize of the separation bound. 
Using fast algorithms \cite{LickteigRoy:FastCauchy:01,Reischert:subresultant:97} 
this cost is $\sOB( d^2( \tau + s))$;
to derive the overall complexity we should multiply it by $\#(T)$.
Notice that for the evaluation we use the sequence of the
quotients, which we computed in $\sOB(d^2 \tau)$
\cite{LickteigRoy:FastCauchy:01,Reischert:subresultant:97}, 
and not the whole Sturm sequence, which can be computed in $\sOB(d^3 \tau)$,
e.g.~\cite{Dav:TR:85}. 

The previous discussion allows us to express the bit complexity of \func{sturm}
not only as a function of the degree and the bitsize, but also using the number
of real roots and the (logarithm of) separation bound. This complexity is
output sensitive, and is of independent interest, although it leads to a
loose worst-case bound.

\begin{lemma}
  \label{lem:sturm-output-sens}
  Let $A \in \ZZ[x]$, $\dg{ A} = d$, $\bitsize{ A} = \tau$ 
  and let $s$ be the bitsize of its separation  bound.
  Using \func{sturm}, we isolate the real roots of $A$ with worst-case
  complexity $\sOB( r d^2 (s^2 + \tau  s))$, where $r$ is the number of real roots.
\end{lemma}

In the worst case $s = \OO(d \tau)$, and to derive the worst case
complexity bound for \func{sturm}, $\sOB( d^4 \tau^2)$,
we should also take into account that $d s = \OO( d \tau)$.

To derive the expected complexity we should consider two cases for the
separation bound, that is, smaller or bigger than $l = 1/ (d^c \tau)$,
where $c \geq 1$ is a small constant that shall be specified later.

In the first case, that is  $\Delta \leq l = 1 / (d^c \tau)$,
the real roots are not well separated, so we rely on the worst
case bound for isolating them, that is $\sOB(d^4 \tau^2)$. 
This occurs with probability 
$\Pr[\Delta \leq l] = \Theta( \sqrt{d} \, l^2) = \Theta(\frac{1}{d^{2c-1/2} \,\tau^2})$,
by the computations of Sec.~\ref{sec:so2-polys} and Sec.\ref{sec:weyl-polys}.
This probability is very small.

For the second case, since  $\Delta > 1 / (d^c \tau)$ we deduce 
$s = \OO(lg{d} + \lg{\tau})$.
The complexity of isolating the real roots, following
Lem.~\ref{lem:sturm-output-sens} is $\sOB( r d^2 \tau)$.
The computations in Sec.~\ref{sec:so2-polys} and
Sec.\ref{sec:weyl-polys} suggest that this case occurs with probability
$\Pr[\Delta > l] = 1 - \Theta(\sqrt{d} \, l^2) =
1 - \Theta(\frac{1}{d^{2c-1/2} \,\tau^2})$, which is close to one.

The expected-case complexity bound of \func{sturm} is 
\begin{displaymath}
  \sOB\Paren{
    (1- \frac{1}{d^{2c-1/2} \,\tau^2}) \cdot rd^2 \tau + 
    \frac{1}{d^{2c-1/2} \,\tau^2} \cdot d^4 \tau^2
} = \sOB( rd^2 \tau),
\end{displaymath}
for any $c\ge 1$, by using $\sqrt{d} =\sO(r\tau)$, which
follows from the expected number of real roots.
To avoid using this expected number, it suffices to set $c\ge 2$.

\begin{theorem}
  \label{th:sturm-avg}
  Let $A \in \ZZ[x]$, where  $\dg{ A} = d$, $\bitsize{ A} = \tau$.
  If $A$ is either a $SO(2)$ or a Weyl random polynomial, then 
  the expected complexity of \func{sturm} solver is
  $\sOB( r \, d^2 \tau)$.
\end{theorem}

In practice, the Sturm sequence is used and not the quotient sequence.
The cost of the former is $\sOB(d^3 \tau)$ which dominates the bound of
Th.~\ref{th:sturm-avg}. This explains the empirical observations that
most of the execution time of \func{sturm} solver is spend on the
construction of the Sturm sequence. 

\section{Random Bernstein polynomials}
\label{sec:random-bernstein}

We compute the expected number of real roots
of polynomials with random coefficients, represented in the
Bernstein basis. We start with some lemmata.

\begin{lemma} 
  \label{lem:binomial-sum}
  For $k \leq n$, non-negative integers, it holds 
  \begin{displaymath}
    \Sum_{j=0}^{n}{ {k \, n \choose k \, j} x^{k \, j}} =
    \frac{1}{k} \, \sum_{j=0}^{k - 1}{ (x + {e}^{\ii \frac{2 \, \pi \, j}{k}})^{k \, n}}.
  \end{displaymath}
\end{lemma}
\begin{proof}
  We consider the RHS of the equality.
  For a specific $j$ we expand the summand, and get terms of the form  
  \begin{displaymath}
    {k \, n \choose \mu} \, x^{kn - \mu} \, e^{\ii \frac{2 \, \pi \, j}{k} \mu},
\; 0 \leq \mu \leq kn. 
  \end{displaymath}
  There are $kn+1$ such terms.
  Recall that $e^{\ii \, 2 \pi} = 1$.
  Let $\mu = \lambda k + \nu$, where $1 \leq \nu \leq k-1$, $0\le\lambda <n$, then
  \begin{displaymath}
    {n k \choose \lambda k + \nu} \, x^{kn -  \lambda k - \nu} \, 
      e^{\ii \frac{2 \, \pi \, j}{k}(\lambda k + \nu)} = 
      {n k \choose \lambda k + \nu} \, x^{kn -  \lambda k - \nu} \, 
      e^{\ii \frac{2 \, \pi \, j}{k}\lambda k} \, 
      e^{\ii \frac{2 \, \pi \, j}{k}\nu}  = 
      {n k \choose \lambda k + \nu} \, x^{kn -  \lambda k - \nu} \, 
      e^{\ii \frac{2 \, \pi \, j}{k}\nu}.
    \end{displaymath}
  If we sum all these terms over $j$, we get
  \begin{displaymath}
    \sum_{j=0}^{k-1}{ {n k \choose \lambda k + \nu} \, x^{kn -  \lambda k - \nu} \,  e^{\ii \frac{2 \, \pi \, j}{k}\nu} } = 
    {n k \choose \lambda k + \nu} \, x^{kn -  \lambda k - \nu} \,
    \sum_{j=0}^{k-1}{ e^{\ii \frac{2 \, \pi \, j}{k}\nu} } =  0,
  \end{displaymath}
  since $\sum_{j=0}^{k-1}{ e^{\ii \frac{2 \, \pi \, j}{k}} } = 0$.

  Let $\mu = \lambda k$. In this case, we have 
  \begin{displaymath}
    {n k \choose \lambda k} \, x^{kn - \lambda k} \, e^{\ii \frac{2 \, \pi \, j}{k}
      \lambda k} = 
    {n k \choose \lambda k} \, x^{kn - \lambda k} = 
    {n k \choose k(n - \lambda)} \, x^{k (n - \lambda)}
  \end{displaymath}
  Notice that $0 \leq \lambda \leq n$.
  Summing up over all $\lambda$ and all $j$, and
  multiplying by $1/k$ we get the LHS.
\end{proof}

\begin{lemma}
  \label{lem:binomial-power}
  For non-negative integers $n, k, p$,
  \begin{displaymath}
    \frac{{n \choose k}^p}{{pn \choose pk}} 
    \approx \sqrt{ p \Paren{ \frac{n}{2 \pi}}^{p-1}} \, \sqrt{ \Paren{\frac{1}{k(n-k)}}^{p-1}}.
  \end{displaymath}
\end{lemma}
\begin{proof}
  The proof follows easily from Stirling's approximation
  $n! \approx \sqrt{2 \pi n}\,(\frac{n}{e})^n$. 
\end{proof}

More accurate results could be obtained if the more precise expression
$\sqrt{2 \pi n}\,(\frac{n}{e})^n \, e^{\frac{1}{12n+1}}< 
n! < \sqrt{2 \pi n}\,(\frac{n}{e})^n \, e^{\frac{1}{12n}}$,
is considered.

\subsection{The expected number of real roots}

We aim to count the real positive roots of a random polynomial in the Bernstein
basis of degree $d$, i.e.
\begin{equation}
  \label{eq:bernstein-poly}
  \widehat{P} := \sum_{k=0}^{k=d} b_k {d \choose k} z^{k}(1-z)^{d-k},
\end{equation}
where we assume that $\widehat{P}(0) \widehat{P}(1) \not= 0$,
and $\set{b_k}$ is an array of random real numbers,
following the normal distribution, with ``moderate'' standard deviation, 
which shall be specified below.

\begin{figure}[t] \centering
  \includegraphics[scale=0.45]{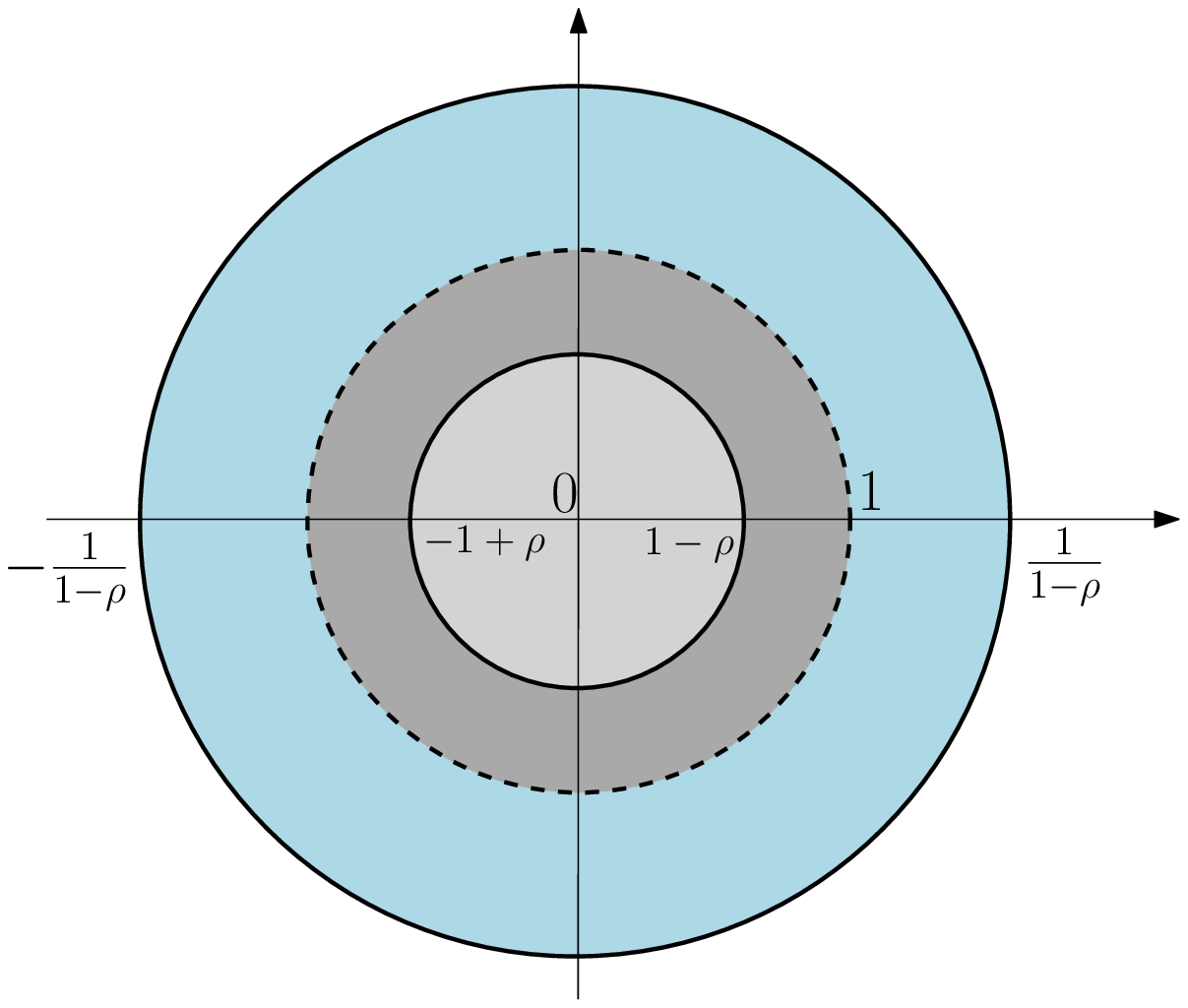} 
  \hspace{30pt}
  \includegraphics[scale=0.45]{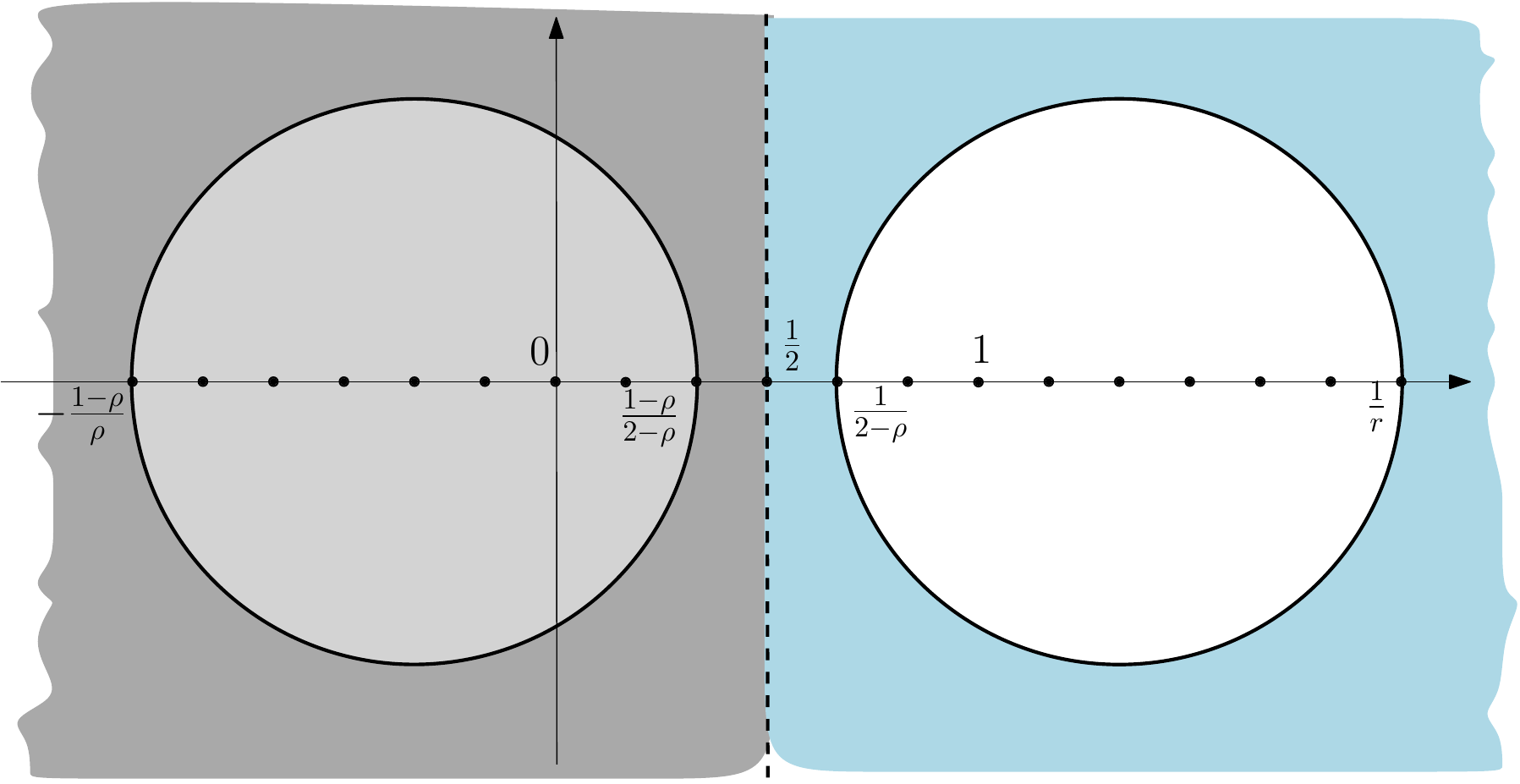}
  \caption{The transformation $z:= \frac{y}{y+1}$ in $\mathbb{C}$.}
  \label{fig:transformation}
\end{figure}

We introduce a suitable change of coordinates, $z := y/(y+1)$,
to transform a polynomial in the Bernstein basis into one in the monomial basis,
by setting $P = (1+y)^d \widehat{P}(y)$.
Now, $P$ and $\widehat{P}$ have the same number of real roots, and 
\begin{displaymath}
  P = \sum_{k=0}^{k=d} b_k {d \choose k} y^k.
\end{displaymath}
Even though the number of real roots does not change, their distribution over
the real axis does, see Fig.~\ref{fig:transformation}.
In particular, we can now apply the techniques already used by Edelman, Kostlan,
and others for counting the number (and, eventually, the limit distribution)
of real roots.
Of course, by symmetry, the expected number of positive and negative
real roots is equal.

By Lem.~\ref{lem:binomial-power}, setting $p=2$ and $n=d$ we deduce:
\begin{equation}
  {d \choose k} \approx \sqrt{ \sqrt{\frac{d}{\pi}} \, \sqrt{\frac{1}{k(d-k)}}} 
  \sqrt{ {2d \choose 2 k}}
  =: \sqrt{S_k}  \,  \sqrt{ {2d \choose 2 k}} .
  \label{eq:binomial-analogy}
\end{equation}
It holds that 
$ \sqrt{2/ {\sqrt{\pi d}}} \leq {{d \choose k}}/{\sqrt{{2d \choose 2k}}} = \sqrt{S_k} \leq 1$. 
To prove this, notice that $S_k$ is decreasing from~1 to $d/2$
and increasing from $d/2$ to $d-1$.
Hence the lower bound is attained at $k=d/2$ and the upper bound at $k=1$ and $k=d-1$.

Since $S_k$ is small compared to ${d \choose k}$, 
it is reasonable to assume that omitting it will
make only a negligible change in the asymptotic analysis. 
 
Let $y = x^2$, with $x > 0$.
Now the problem at hand to count the positive real roots of
\begin{displaymath}
  P = \sum_{k=0}^{k=d}{ a_k \sqrt{ {2d \choose 2k}} x^{2k}}.
\end{displaymath}
We need the following proposition
\begin{proposition} \label{prop:EK} \cite{EdeKos-BAMS-95}
  Let $v(t) = (f_o(t), \dots, f_n(t))^{\top}$ be a vector of differentiable
  functions and $c_0, \dots, c_n$ elements of a multivariate normal distribution with
  zero mean and covariance matrix $C$. The expected number of real zeros on an interval
  (or a measurable set) $I$ of the equation
  $c_0 f_0(t) + \cdots + c_n f_n(t) = 0$,
  is
  \begin{displaymath}
\int_I{ \frac{1}{\pi} \norm{ \bvec{w}'(t)} dt},\; \bvec{w} = w(t)/ \norm{w(t)}.
  \end{displaymath}
  where $w(t) = C^{1/2} v(t)$.
  In logarithmic derivative notation, this is 
  \begin{displaymath}
    \frac{1}{\pi} \int_I{ \sqrt{ \frac{\partial^2}{\partial x \partial y}
    \log{ (v(x)^{\top} C v(x)}) |_{x=y=t}} dt}.
  \end{displaymath}
\end{proposition}

For computing the integral in Prop.~\ref{prop:EK}, we shall use
the logarithmic derivative notation.
Following Prop.~\ref{prop:EK},
$f_{2i}(t) = \sqrt{{2d \choose 2i}}x^{2i}$ and $f_{2i+1}(t) = 0$, 
$c_{2i} = a_i$ and $c_{2i+1}=0$, where $0 \leq i \leq d$,
and the variance is~1.
Then,
\begin{displaymath}
  v(x)^{\top} C v(y) = \sum_{k=0}^{d} {2d \choose 2k} (xy)^{2k} .
\end{displaymath} 
We consider function
\begin{displaymath}
  f(z) = \sum_{k=0}^{d} {2d \choose 2k} z^{2k}.
\end{displaymath}
By Lem.~\ref{lem:binomial-sum}, for $k=2$, we have
$f(z) = \frac{1}{2} \Paren{ (1 + z)^{2d} + (1 -z)^{2d}}$ 
and so
$f^{'}(z) = d( z+1)^{2d-1} + d( z-1)^{2d-1}$,
$f^{''}(z) = d(2d-1)( z+1)^{2d-2} + d(2d-1)( z-1)^{2d}$.

The following  quantities are also relevant
$f f^{'} = \frac{1}{2}d( z+1)^{4d-1} + dz(z^2-1)^{2d-1} + \frac{1}{2}( z-1)^{4d-1}$,
$f f^{''} = \frac{1}{2}d(2d-1)(z+1)^{4d-2} + d(2d-1)(z^2+1)(z^2-1)^{2d-2} +
\frac{1}{2}(2d-1)( z-1)^{4d-2}$,
and $ (f')^{2} = d^2(z+1)^{4d-2} + 2d^2(z^2-1)^{2d-1} + d^2(z-1)^{4d-2}$.

It holds that 
\begin{displaymath}
  \partial_x \partial_y( \log f( x, y)) = 
  \frac{f^{'} f + x y f^{''}f - xy (f^{'})^2}{f^2} =
  \frac{A}{f^2},
\end{displaymath}
with
\begin{displaymath}
\hspace{-0.2cm}
  \begin{array}{rl}
    A =  & d(z+1)^{4d-2}( \frac{1}{2}(z+1) + \frac{1}{2}(2d-1)z  - zd) \\
        & + d(z^2-1)^{2d-2}( z(z^2-1) + (2d-1)z(z^2+1) - 2d(z^2-1)z) \\
        & -d(z-1)^{4d-2}(\frac{1}{2}(z-1) + \frac{1}{2}(2d-1)z -zd) \\
      = & \frac{1}{2} d \left( (z+1)^{4d-2} + 4(2d-1)z(z^2-1)^{2d-2} - (z-1)^{4d-2} \right).
  \end{array}
\end{displaymath}
If we let $z = t^2$, then
\begin{displaymath}
\hspace{-0.2cm}
  \begin{array}{ll}
    \frac{A(t^2)}{f(t^2)^2} & 
    = \frac{\frac{1}{2}d ((1+t^2)^{4d-2} + 4(2d-1)t^2(t^4-1)^{2d-2} - (t^2-1)^{4d-2} ) }
    {\frac{1}{4}( (1+t^2)^{2d} + (1-t^2)^{2d})^2} \\
    & = 2d \frac{1}{(1+t^2)^2}
    \frac{1 + (2d-1) \left( \frac{2t}{1+t^2} \right)^2 \left(\frac{1-t^2}{1+t^2} \right)^{2d-2} - \left(\frac{1-t^2}{1+t^2} \right)^{4d-2}}
    {\left( 1 + \left( \frac{1-t^2}{1+t^2} \right)^{2d} \right)^2}.
  \end{array}
\end{displaymath}

We consider the substitutions
$t = \tan \frac{\theta}{2}$,
$\tan \theta = \frac{2t}{1-t^2}$,
$\sin \theta = \frac{2t}{1+t^2}$,
$\cos \theta = \frac{1-t^2}{1-t^2}$, 
and $\frac{d \theta}{2} = \frac{dt}{1+t^2}$.
Then
\begin{displaymath}
  \begin{array}{l}
  \frac{A}{f(t^2)^2} = 2d \frac{1}{(1 +t^2)^2}
  \frac{1 + (2d-1) \sin^2{ \theta} (\cos{ \theta})^{2d-2} - (\cos{ \theta})^{4d-2}}
       { \left( 1 + (\cos{ \theta})^{2d} \right)^2}.
     \end{array}
\end{displaymath}

The expected number of positive real roots is given by
\begin{displaymath}
  \begin{array}[]{lll}
  I & = & \frac{1}{\pi} \int_{0}^{\infty}{ \frac{\sqrt{A}}{f(t^2)}} d t \\
    & = &\frac{1}{\pi} \int_{0}^{\pi}{ 
    \sqrt{2d} \frac{\sqrt{1 + (2d-1) \sin^2 \theta (\cos \theta)^{2d-2} - (\cos \theta)^{4d-2} }}
                   {1 + (\cos \theta)^{2d}}
     } \frac{d \theta}{2}.
   \end{array}
\end{displaymath}

Performing the change $\theta \mapsto \pi - \theta$, we notice
that $I$ equals twice the integral between $0$ and $\pi/2$. 
Hence, the expected number of positive real roots of $P$ in $(0, 1)$ equals
that in $(1, \infty)$. Hence,
\begin{displaymath}
  \begin{array}{l}
  I = \frac{ \sqrt{2d}}{\pi} \int_{0}^{\pi/2}{ 
    \frac{ \sqrt{ 1 + (2d-1) \sin^2 \theta (\cos \theta)^{2d-2} - (\cos \theta)^{4d-2} }}
                   {1 + (\cos \theta)^{2d}}
     } d \theta.
   \end{array}
\end{displaymath}

Now we will bound the integral as $d \rightarrow \infty$.
Applying the triangular inequality and noticing that
$1 + (cos \theta)^{4d-2} \leq 1$ and $a + \cos{\theta}^{2d} \geq 1$,
we get
\begin{displaymath}
  \begin{array}{lcl}
    I & \leq & \frac{\sqrt{ 2d}}{\pi} \left[ 
      \int_0^{\pi/2}{\sqrt{1}\, d \theta} + 
      \int_0^{\pi/2}{\sqrt{2d-1} \sin{\theta} (\cos{\theta})^{d-1}\, d \theta} 
    \right] \\
    & = &  \frac{\sqrt{ 2d}}{\pi} 
    \left( \frac{\pi}{2} + \sqrt{2d-1}\, \frac{1}{d} \right) = 
    \frac{\sqrt{ 2d}}{2} 
    \left( 1 + \frac{1}{\pi} \frac{\sqrt{2d-1}}{d} \right) \\
    & \leq &  \frac{\sqrt{ 2d}}{2} \left( 1 + \frac{1}{\pi}
      \sqrt{\frac{2}{d}}\right) 
    \leq \frac{\sqrt{ 2d}}{2} + \frac{1}{\pi}.
  \end{array}
\end{displaymath}

For a lower bound, we neglect the positive term
$(2d-1) \sin^2 \theta (\cos \theta)^{2d-2}$, and notice that 
$\sqrt{ 1 + (\cos{ \theta})^{4d-2}} \geq 1 +(\cos{ \theta})^{2d-1} 
\geq  1 +(\cos{ \theta})^{2d-2} = (1 +(\cos{ \theta})^{d-1})(1 -(\cos{ \theta})^{d-1})$,
and
$\frac{1 +(\cos{ \theta})^{d-1}}{1 +(\cos{ \theta})^{2d}} \geq 1$.

\begin{lemma}
  $$
  W(n) := 
  \int_{0}^{\pi/2}{ ( \cos{ \theta})^{n} \, d \theta} 
  \leq 
  \frac{2}{\sqrt{\pi}} \frac{1}{\sqrt{d+1}}.
  $$
\end{lemma}
\begin{proof}
  We need the following inequality \cite{cq-ams-2004}
  on Wallis' cosine formula:
  \begin{displaymath}
    \frac{1}{\sqrt{\pi(k + 4\pi^{-1} -1)}}
    \leq
    \frac{1 \cdot 3 \cdot 5 \cdots (2k-1)}{ 2 \cdot 4 \cdot 6 \cdots (2k)}
    \leq
    \frac{1}{\sqrt{\pi(k+1/4)}}.
  \end{displaymath}
  
  If $n$ is even then
  $W(n) =
  \frac{\pi}{2} \frac{(n-1)!!}{n!} =
  \frac{\pi}{2} \frac{1\cdot3\cdot5 \cdots(n-1)}{2\cdot4\cdot6\cdots(n)}
  \leq \frac{\pi}{\sqrt{\pi(2n+1)}}
  \leq \frac{2}{\sqrt{\pi}} \frac{1}{\sqrt{d+1}}
  $.

  If $n$ is odd then 
  $W(n) = 
   \frac{(n-1)!!}{n!} =
   \frac{2\cdot4\cdot6\cdots(n-3)(n-1)}{1\cdot3\cdot5 \cdots(n-2)n}
   \leq
   \sqrt{\pi(k + 4\pi^{-1} -1)} \cdot \frac{1}{n}
   \leq \frac{2}{\sqrt{\pi}} \frac{1}{\sqrt{d+1}}
  $.
\end{proof}
Using the lemma,
$\int_{0}^{\pi/2}{ ( \cos{ \theta})^{d-1} \, d \theta} 
\leq \frac{2}{\sqrt{\pi}} \frac{1}{\sqrt{d}}$, so:
\begin{displaymath}
  \begin{array}{lll}
    I & \geq &
    \frac{\sqrt{ 2d}}{\pi} \int_0^{\pi/2}{1 - (\cos{ \theta})^{d-1} \, d \theta} \\
    & \geq & \frac{\sqrt{ 2d}}{2} \left( 1 - \frac{4}{\pi \sqrt{ \pi}} \frac{1}{\sqrt{d}} \right)
     \geq  \frac{\sqrt{ 2d}}{2} - \sqrt{ \frac{8}{\pi^3}}.
  \end{array}
\end{displaymath}
Hence $I = \frac{\sqrt{ 2d}}{2} \pm \OO( 1)$ and we can state the following:

\begin{theorem}
  \label{th:nb-roots-2d-2k}
  The expected number of real roots of a random polynomial
  $P = \sum_{k=0}^{k=d}{ a_k \sqrt{ {2d \choose 2k}} x^{2k}}$,
  where $a_k$ are standard normals i.i.d. random variables,
  is $\sqrt{ 2d} \pm \OO( 1)$.
\end{theorem}

By employing (\ref{eq:binomial-analogy}) and considering $\sqrt{S_k}$ as
part of the deviation, we have the following:
\begin{corollary}
  \label{cor:nb-roots-bern-sk}
  The expected number of real roots of a random polynomial in the
  Bernstein basis, Eq.~(\ref{eq:bernstein-poly}),
  the coefficients of which are normal i.i.d. random variables with mean 0
  and variance $1/S_k = 1/\sqrt{\frac{d}{\pi k (d-k)}}$,
  is  $\sqrt{ 2d} \pm \OO( 1)$.
\end{corollary}

In Table~\ref{tab:rand-bern} we present the results of experiments with
polynomials in the Bernstein basis (see Eq.~(\ref{eq:bernstein-poly})), 
of degree $\le 1000$,
the coefficients of which are i.i.d.\ random variables following the standard
normal distribution, 
that is  mean zero and variance 1.
For each degree we tested 100 polynomials.
The first column is the degree, while the second is the expected
number of real roots predicted by Cor.~\ref{cor:nb-roots-bern-sk} which
assumes variance $1/S_k$.
The third column is the average number of real roots computed.
Our experiments support the following conjecture:
\begin{conjecture}
  \label{cl:nb-roots-bern}
  The expected number of real roots of a random polynomial in the
  Bernstein basis, Eq.~(\ref{eq:bernstein-poly}),
  the coefficients of which are standard normal i.i.d. random
  variables,
  that is with mean 0 and variance 1, is $\sqrt{ 2d} \pm \OO( 1)$.
\end{conjecture}

Columns 4-7 of Tab.~\ref{tab:rand-bern} corresponds to the average
number of real roots in the intervals $(-\infty, -1)$, $(-1, 0)$, 
$(0,1)$ and $(1, \infty)$, respectively.
For these experiments we took random polynomials in the monomial basis
and converted them to the Bernstein basis.
The roots of a random polynomial in the monomial basis, under the assumptions of
\cite{HugNik-cm-2008}, concentrate around the unit circle.
The symmetry of the density suggests that each of the intervals 
$(-1/(1-\rho), -1)$, $(-1, -1 + \rho)$, $(1 - \rho, 1)$, and $(1, 1/(1-\rho))$,
contains on the average $1/4$ of the real roots
(Fig.~\ref{fig:transformation}, left).
If we apply the transformation $x \mapsto x/(x+1)$
(Fig.~\ref{fig:transformation}, right)
to transform the polynomial to the Bernstein basis, 
then $3/4$ of the real roots are positive, 
$1/2$ of them are in $(0, 1)$ and $1/4$ in $(1, \infty)$.
We refer to the last columns of Tab.~\ref{tab:rand-bern} for
experimental evidences of this.

\begin{figure}[t] \centering
  \includegraphics[scale=0.25]{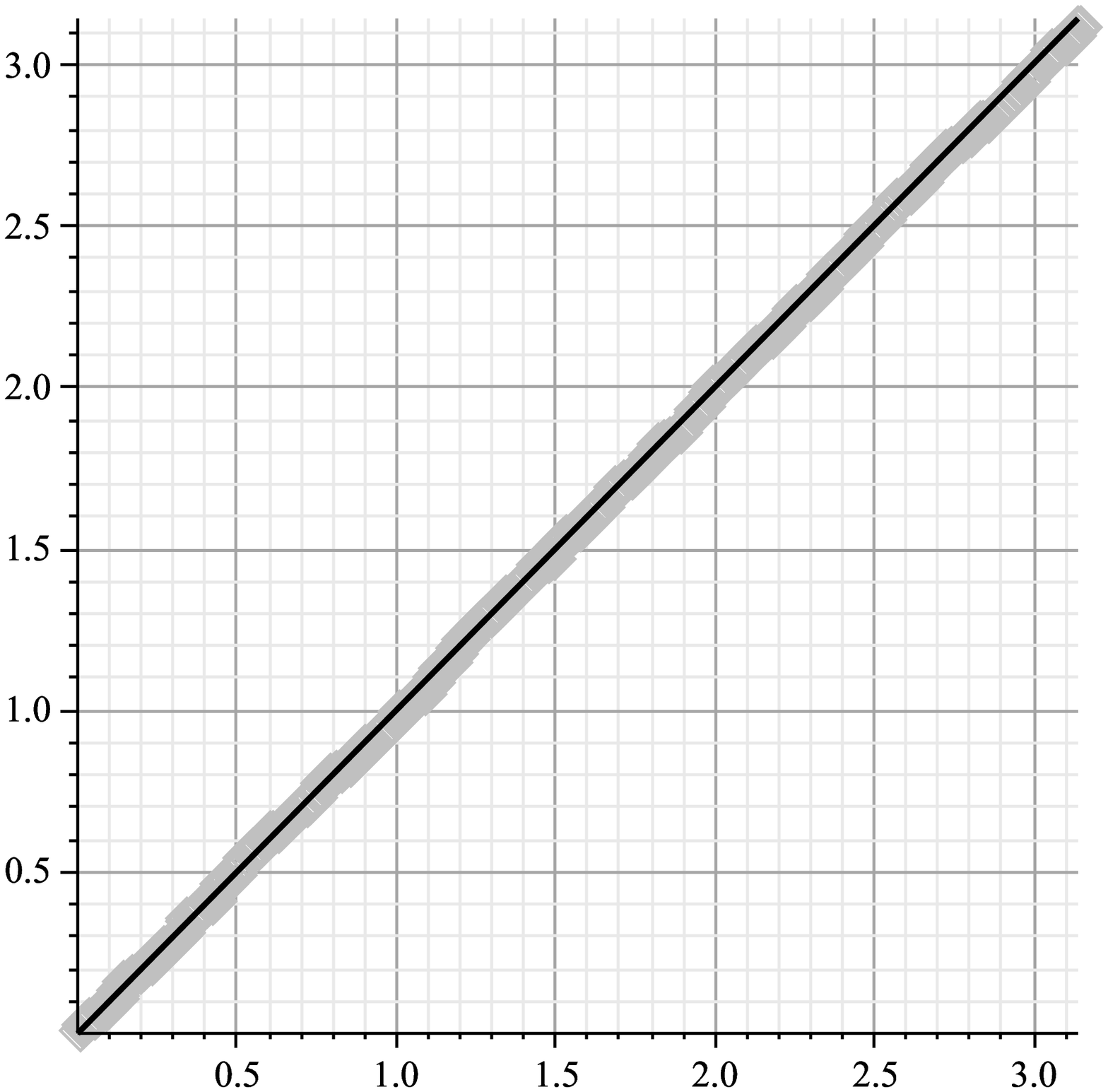} \hspace{30pt}
  \includegraphics[scale=0.25]{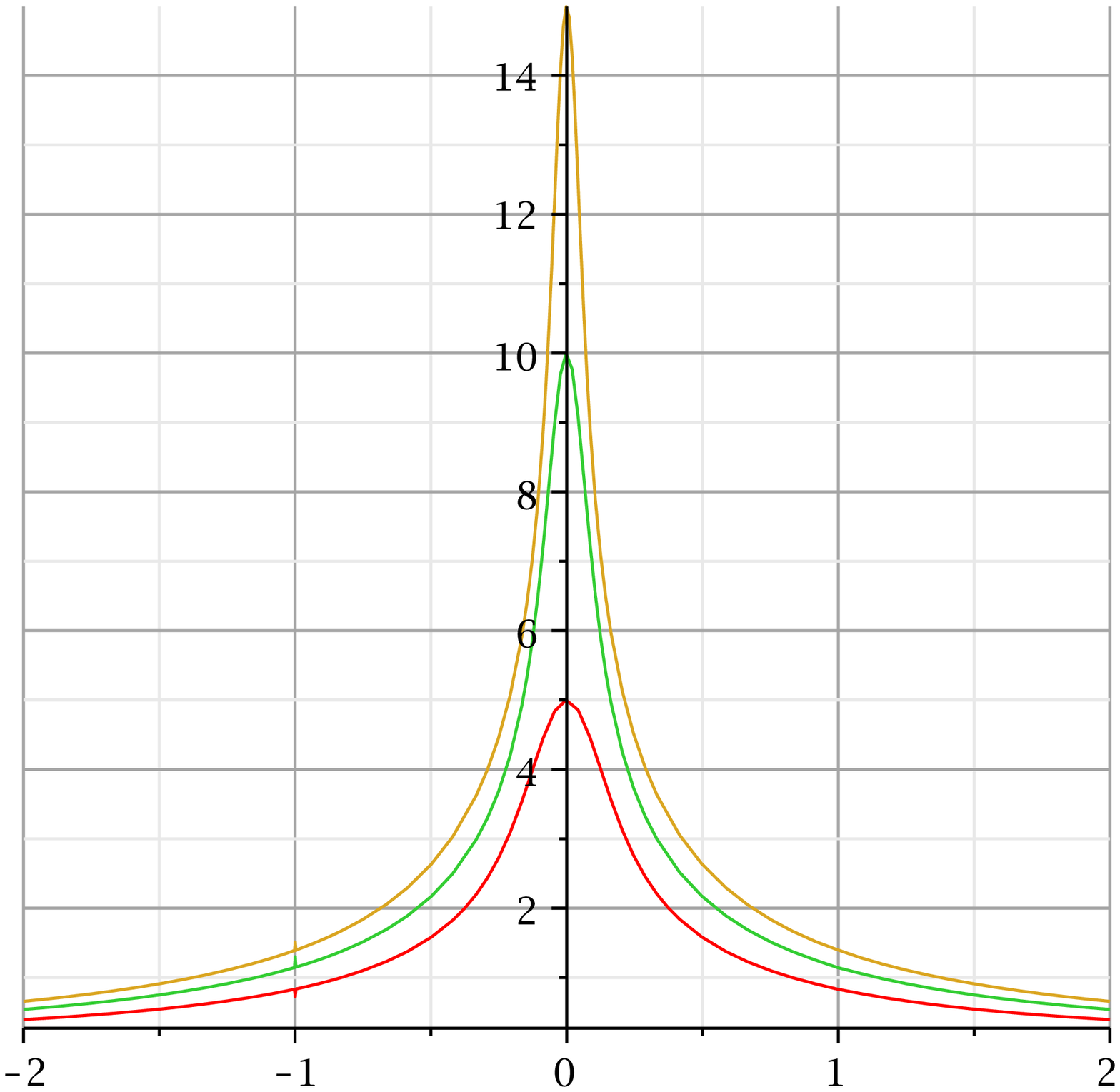}
  \caption{Left: Function $\arccos(2t-1)$ of real roots in $(0, 1)$, against uniform
    distribution in $(0, \pi)$. 
    Right: Density of polynomials in the Bernstein basis for $d \in \{5,10,15\}$.}
  \label{fig:roots-in-0-1}
  \label{fig:density-bernstein}
\end{figure} 

As far as the distribution of the real roots in $(0, 1)$ is concerned,
if we denote them by $t_i$, then $\arccos(2 t_i-1)$, behaves as the uniform
distribution in $(0, \pi)$.
In Fig.~\ref{fig:roots-in-0-1}, we present the probability-probability plot,
(using the \texttt{ProbabilityPlot} command of \func{maple}) of
this function of real roots of random polynomials in Bernstein basis, of degree 1000 
(light grey line), against the theoretical uniform distribution (black line) in
$(0, \pi)$.
We observe that the lines almost match.
For reasons of space, we postpone the discussion about the distribution of the roots.

\begin{table}[t]
  \small
  \centering
  \begin{tabular}{|c||c||c||c|c|c|c|}
    \hline
    {\tiny $d$} & {\tiny $\sqrt{2d}$} & {\tiny $(-\infty, \infty)$} & {\tiny $(-\infty, -1)$} 
    & {\tiny $(-1, 0)$} & {\tiny $(0, 1)$} & {\tiny $(1, \infty)$} \\ \hline 
    100  &  14.142 &   13.640 &   0.760 &   2.740 &   6.530 &   3.610 \\ \hline 
    150  &  17.321 &   16.540 &   0.890 &   3.260 &   8.090 &   4.300 \\ \hline 
    200  &  20.000 &   19.740 &   1.100 &   3.780 &   9.740 &   5.120 \\ \hline 
    250  &  22.361 &   21.400 &   1.350 &   3.970 &  10.610 &   5.470 \\ \hline 
    300  &  24.495 &   24.320 &   1.270 &   4.760 &  12.300 &   5.990 \\ \hline 
    350  &  26.458 &   26.540 &   1.620 &   5.100 &  13.400 &   6.420 \\ \hline 
    400  &  28.284 &   27.980 &   1.490 &   5.430 &  14.080 &   6.980 \\ \hline 
    450  &  30.000 &   29.460 &   1.620 &   5.890 &  14.970 &   6.980 \\ \hline 
    500  &  31.623 &   31.200 &   1.830 &   5.960 &  15.620 &   7.790 \\ \hline 
    550  &  33.166 &   32.740 &   1.770 &   6.360 &  16.290 &   8.320 \\ \hline 
    600  &  34.641 &   34.300 &   1.850 &   6.570 &  17.270 &   8.610 \\ \hline 
    650  &  36.056 &   35.480 &   2.050 &   6.840 &  17.240 &   9.350 \\ \hline 
    700  &  37.417 &   37.200 &   2.160 &   7.510 &  18.650 &   8.880 \\ \hline 
    750  &  38.730 &   38.180 &   2.190 &   7.300 &  19.360 &   9.330 \\ \hline 
    800  &  40.000 &   39.160 &   2.220 &   7.830 &  19.490 &   9.620 \\ \hline 
    850  &  41.231 &   40.420 &   2.130 &   8.010 &  20.320 &   9.960 \\ \hline 
    900  &  42.426 &   41.780 &   2.390 &   8.070 &  20.530 &  10.790 \\ \hline 
    950  &  43.589 &   42.680 &   2.200 &   8.330 &  21.570 &  10.580 \\ \hline 
    1000  &  44.721 &   43.540 &   2.400 &   8.610 &  21.770 &  10.760 \\ \hline 
  \end{tabular}
  \caption{Experiments with random polynomial in the Bernstein basis.}
  \label{tab:rand-bern}
\end{table} 

\section{Conclusions and future work}

Our results explain why the solvers are fast in general, since
typically there are few real roots and in general the separation bound
is good enough.
This agrees with the fact that in most cases the practical complexity of the
\func{sturm} solver is dominated by the computation of the sequence
and not by the evaluation.  
Our current work extends the first part of this paper to Kac polynomials,
and to solvers \func{descartes} and \func{bernstein}.

The main issue with the Kac polynomials is that there is a
discontinuity at $\pm 1$ when $d \rightarrow \infty$.
To be more precise, the fact that there are few roots even near $\pm 1$,
where they are concentrated asymptotically, is balanced
by the fact the 2-point correlation, $k(s_1, s_2)$,  
between two consecutive roots is a complicated
function of $|s_1-s_2|$, $s_1$ and $d$ and (in opposition with
the two other distributions we studied) its limit
when $d$ tends to infinity
is not equivalent to a simple function of $|s_1-s_2|$.
This is an interesting  problem which deserves to be 
studied and investigate further.

An interesting question is whether we can design a randomized exact algorithm
based on the properties of random polynomials.
Lastly, we wish to extend our study to polynomials with inexact coefficients.

{\bf Acknowledgement.}
We thank the reviewers for their constructive comments.
IZE  thanks D.~Hristopoulos for discussions on
the statistics of roots' distributions.
AG acknowledge fruitful discussions with Julien Barre. 
IZE and AG are partially supported by 
Marie-Curie Network ``SAGA'', FP7 contract PITN-GA-2008-214584.
ET is partially supported by an individual postdoctoral grant from the
Danish Agency for Science, Technology and Innovation,
and by the State Scholarship Foundation of Greece.

{ 
\bibliographystyle{plain}
\bibliography{random}
}

\end{document}